\documentclass[]{article}

\usepackage{lmodern}
\usepackage{xcolor}
\usepackage[margin=1.0in]{geometry}
\usepackage{amssymb,amsmath}
\usepackage{parskip}
\usepackage{pifont}

\usepackage{algpseudocode}
\usepackage{float}
\usepackage{ifxetex,ifluatex}
\usepackage{optidef}
\usepackage[htt]{hyphenat}
\usepackage[colorlinks=true,
            linkcolor = black,
            urlcolor  = blue,
            citecolor = blue,
            anchorcolor = black]{hyperref}
\usepackage[utf8]{inputenc}
\usepackage{xcolor}
\usepackage{MnSymbol}
\usepackage{amsmath}
\usepackage{mathtools}
\usepackage{authblk}
\usepackage{amsfonts}
\usepackage{enumerate}
\usepackage{amsthm}
\usepackage{comment}
\usepackage{dsfont}
\usepackage{pgfplots}
\usepackage{biblatex}
\usepackage{pgf-pie} 
\usepackage{algorithm}
\usepackage{algpseudocode}
\usepackage{tikz}
\usepackage{caption}
\usepackage{subcaption}
\usepackage{tikz-cd}
\usepackage{tabulary}
\usepackage{booktabs}

\usepackage[skins,breakable]{tcolorbox}
\graphicspath{ {./images/} }
\usepackage{float}
\newtcolorbox{mybox}[1][]{enhanced jigsaw,breakable,pad at break=1mm,
  oversize,left=8mm,interior hidden,colframe=black,nobeforeafter=,#1}

\usepackage[edges]{forest}

\theoremstyle{definition}
\newtheorem{theorem}{Theorem}[section]
\newtheorem{lemma}[theorem]{Lemma}
\newtheorem{definition}[theorem]{Definition}

\newtheorem{examples}[theorem]{Examples}
\newtheorem{corollary}[theorem]{Corollary}
\newtheorem{proposition}[theorem]{Proposition}

\begin{document}

\author[1]{Bruno Mazorra}
\author[2]{Nicol\'as Della Penna}
\affil[1]{Universitat Pompeu Fabra}
\affil[2]{Amurado Research}

\title{Constant Function Market Making, Social Welfare and Maximal Extractable Value\footnote{Work in progess.}}

\maketitle

\begin{abstract}
We consider the social welfare that can be facilitated by a constant function market maker (CFMM).
When there is sufficient liquidity available to the CFMM, it can approximate the optimal social welfare when all users transactions are executed. 
When one of the agent has the role of proposing the block, and blockspace is scarce, they can obtain higher expected utility than otherwise identical agents. This gives a lower bound on the maximal extractable value exposed when blockspace is scarce.
\end{abstract}

\section{Introduction}

Constant function market makers (CFMMs) are one of the leading application of distributed consensus systems. These markets accurately report prices under the existence of off-chain markets and non-arbitrage condition \cite{angeris2021constant}. In this paper, we explore the properties of the CFMM to aproximate Walrassian Equilibrium prices  when players act sequentially with Walrasian demands. Moreover, we study welfare in a model exchange economy with random trader endowments, similar to that of \cite{frongillo2012interpreting} use to study automated market makers (AMMs) for securities with binary payoffs, motivated by prediction markets. 
We show that if the liquidity available to the CFMM asymptotically increases relative to the wealth of the traders arriving in a single period, 
it can approximate the optimal social welfare when all users transactions are executed. 
When one of the agent plays the role of  proposing a block, and thus has the ability to censor transactions, they can obtain higher utility than otherwise identical agents when blockspace is scarce relative to transactions. 
\section{Preliminaries}
\subsection{Welfare in Walrasian equilibrium}\label{one-shoot}
Assume that all players share the same concave utility function $U:\mathbb R^l_+\rightarrow \mathbb R$. Each player have a vector of endowments $\Delta\in\mathbb R^l_+$ drawn from a distribution $\mathcal D$. All players are \textit{utility maximizing price takers}, thus given the  price of endowments is $p$, a player with endowment $\Delta$ solves:
\begin{maxi*}|s|
{\Delta'}{U(\Delta')}
{}{}
\addConstraint{p\cdot \Delta=p\cdot \Delta'.}{}
\end{maxi*}
We denote by $\zeta(\Delta,p)$ the $\Delta'$ that optimizes the problem and by $z(\Delta,p):=\zeta(\Delta,p)-\Delta$. Then the \textbf{Walrasian equilibrium price} is defined as a price vector $p^\star$ such that
\begin{equation*}
    \mathbb E_{\Delta\sim \mathcal D}[z(\Delta,p^\star)]=0.
\end{equation*}
By Arrow-Debreu \cite{mas1995microeconomic}, we have that if $U$ is strictly increasing, convex and the support of the endowments is bounded, then the Walrasian equilibrium exist.
In this paper, we will constrain to the exchange economies with a unique Walrasian equilibrium. For a given set of endowments and a utility function $U$, the \textbf{social welfare} in this game is defined as $Wf(\mathcal D,U):=\mathbb E_{\Delta\sim\mathcal D}[\zeta(\Delta,p^\star)]$.

Now, assume that the number of player instances is a finite number $n$. In this case, in order to find the Pareto optimal allocation, we will do the following.
Consider a case with traders $\mathcal I = \{1,...,n\}$ and assets $\textbf T=\{1,...,l\}$. 
Each player, has a preference over the assets modeled by a utility function $U:\mathbb R^l_+\rightarrow\mathbb R$ and is endowed with a non-negative vector of $l$ goods $\Delta^i=(\Delta^i_1,...,\Delta^i_l)$ drawn from the distribution $\mathcal D$.
Altogether, $\mathcal E=(\mathcal I,\textbf T,U,\Delta=\{\Delta^i\}_i)$ defines an exchange economy.

Without assuming inventory or storage (such as that of an AMM), the first requirement is that the assignment of goods to individuals not exceed the amounts available. We define the allocation of assets as a vector $x=(x^1,...,x^n)$ 
where $x^i=(x^i_1,...,x^i_l)$ denotes trader $i'$s bundle according to the allocation. The set of feasible allocation is defined as 
\begin{equation*}
    F(\Delta)= \{(x^1,...,x^n):\sum_{i\in\mathcal I} x^i=\sum_{i\in\mathcal I} \Delta^i \}
\end{equation*}
and it contains all allocations of goods across individuals that, in total, exhaust the available amount of every good. 
Under sufficient good conditions, there exist a Walrasian equilibrium, and we can compute the Walrasian equilibrium allocations (WEA).

Now let $\Delta_1,....,\Delta_n$ be i.i.d with distribution $\mathcal D$. For each instance of $\Delta =(\Delta_1,...,\Delta_n)$, we can compute the Walrasian equilibrium price $p$. 
And so, the Walrasian equilibrium price $p=p(\Delta_1,...,\Delta)$ follows a distribution $X_n$ with $\text{supp } X_n\subseteq \{x\in\mathbb R^l_+:\sum_{i=1}^l x_i=1\}$. In this paper, we will assume that for a pure exchange economy with $n$ players and endowment distribution $\mathcal D$ with bounded support, we have that $\mathbb E[X_n]=p^\star$ and
$X_n\underset{n\rightarrow+\infty}{\longrightarrow} p^\star$ almost surely.

\subsection{Constant function market makers}

A constant function market maker (CFMM,see \cite{angeris2020improved}) consist of a function $C:\mathbb R^l_+\rightarrow \mathbb R$, its reserves $ R\in\mathbb R^l_+$ and transaction fees parameter $\gamma$. The element $R_i$ is the amount of $i$ assets available on the CFMM contract, while the function $C$ specifies the behavior of the contract. More specifically, if an agent wants to trade with the CFMM a vector of assets $\Lambda\in\mathbb R^l$ if
\begin{equation*}
    C(R+(1-\gamma)\Lambda)\geq C(R).
\end{equation*}
We say that the CFMM has no fees if $\gamma=0$. In this paper, we will assume fee-less CFMM except stated otherwise. The agents that add and remove liquidity on the CFMM are called \textit{liquidity providers}. When adding and removing some reserves by these players, the marginal price of the CFMM can not change, that is, $\nabla C(R')=\nabla C(R)$, see \cite{angeris2021constant}.
\begin{examples} In the following, we provide an incomplete list of examples of CFMM:
\begin{itemize}
    \item Uniswap V2 DEX has the CFMM defined as $C(x,y)=xy$. If the tuple of initial reserves is $R=(R_0,R_1)$ and a trader wants to exchange some amount $\Delta$ of tokens $X$, she obtains $g(\Delta)=-\frac{R_0R_1}{R_0+\Delta}+R_1$.
    \item Constant sum market makers $C(x)=\textbf{c}^t\cdot x$, for some $c\in\mathbb R^l_+$. 
    \item Constant geometric mean market maker $C(x)=\prod_{i=1}^lx_i^{w_i}$, where $w>0$ and $\textbf{1}^t\cdot w=1$.
    \item Constant min market maker $C(x)=\min\{x_1,...,x_l\}$. Observe that if a player with strictly increasing utility function $U$ trades with this market maker will reach a state with reserves of the form $(R_i,...,R_i)$ for some $i$.
    \item Quadratic-over-linear constant market maker $C(x,y)=-x^2/y$ defined in $\{(x,y):y>0\}$.
    \item The Minecraft modification package market maker $C(x,y)=xe^y$.
\end{itemize}
\end{examples}
In this paper, we will not assume that players have a specific utility function. Therefore, we will need to state the \textbf{general trade choice problem}, similar to the problem stated in \cite{angeris2021constant}. We have a player with strictly increasing and concave utility function $U$ with initial endowments $\Delta$. Then, a utility maximizing players that trades with a CFMM with constant function $C$ and reserves $R$ solves the problem
\begin{equation*}
\begin{aligned}
& \underset{\Lambda}{\text{maximize}} & & U(\Lambda) \\
& \text{subject to} & & C(R+\Delta-\Lambda)\geq C(R),\Lambda\geq0
\end{aligned}
\end{equation*}
This problem is a convex problem, thus we can globally and efficiently solve this problem \cite{angeris2021constant}. It is easy to show what the solution will satisfy $C(R+\Delta-\Lambda)=C(R)$. For a player with utility function $U$ and endowments $\Delta$ and a CFMM with function $C$ and reserves $R$, we denote the solution of the optimization problem as $\zeta(U,\Delta,C,R)$ or $\zeta(\Delta,R)$ if $U$ and $C$ are clearly specified. Observe that the solution of the optimization problem is not necesrly unique and therefore $\zeta(\Delta,R)$ is not specified. However, we will see that in some conditions, choosing a solution for a specific $\Delta_0,R_0$ the function $\zeta$ can be extended smoothly.

\subsection{Maximal extractable value}

Maximal (also miner) extractable value, or MEV, usually refers to the value that \textit{privileged} players can extract by strategically ordering, censoring, and placing transactions in a blockchain. In our context, this privileged players will be builders or Walrasian auctioneers, i.e. players responsible for constructing the new block or more generally responsible for allocating the goods an input of endowments and preferences revealed. 

A formal definition of MEV can be seen in \cite{babel2021clockwork,mazorra2022price}. In this paper, we will define the maximal extractable value as follows.
\begin{definition} Let $P$ be a builder with utility function $U_P$ and endowments $\Delta$. Let $\mathcal M$ be an allocation mechanism and $\mathcal T$ the set of transactions received by the builder. Then, we define the Walrasian MEV as
\begin{align*}
& \underset{\mathcal T'}{\text{maximize }}   U_P(\mathcal M(\mathcal T')_P) \\
& \text{subject to } \mathcal T'\subseteq \mathcal T
\end{align*}

\end{definition}

\section{Constant function market makers with Walrasian demand}
In this section, we study the social welfare when players interact with constant function market makers sequentially. We assume that there exist an exogenous liquidity provider with initial liquidity $R_0$ provided to a constant function market maker with convex function $C:\mathbb R^l_+\rightarrow \mathbb R$. All players share the same utility function $U$ and the endowments are drawn from a distribution $\mathcal D$. Except stated otherwise, we will assume that $C$ and $U$ are smooth maps.

\subsection{Model}

In each time period $t$, an agent with endowment $\Delta_t\sim \mathcal D$ and convex, strictly increasing utility function $U:\mathbb R^l_+\rightarrow \mathbb R$ trades in the constant function market maker. After a feasible trade $\Lambda_t$, the liquidity is updated to $R_{t+1}=R_t+\Lambda_t'$. Assuming that all players are \textit{utility maximizing}, we have that, at time $t$ the player solves the general trade choice problem computing $\zeta(\Delta_t,R_t)$.
For a specific realization of endowments $\Delta_1,...,\Delta_t$, we denote $\Delta(t)=(\Delta_1,...,\Delta_t)$. 

To define the notion of social welfare and optimal welfare in equilibrium, one need to define a notion of the equilibrium on the game. One could be tempted to define a generalization of the notion of Walrasian equilibrium provided in section \ref{one-shoot}. Equivalently, one could say that a reserve $R^\star$ in the feasible set of a CFMM with curve $C$ is a Walrasian equilibrium for CFMM if 
\begin{equation}\label{eq:social_bad}
    \mathbb E_\Delta [z(\Delta,R^\star,C)]=0
\end{equation}
However, in general, there is no solution $R^\star$ for the equation \ref{eq:social_bad}, see \ref{appendix} for counterexample. So,
to counterfactual compare different social welfare CFMM, we define the following exchange economy and notion of social welfare.
\begin{definition} A CFMM with Walrasian demands (CFMMWD) consists of a tuple $G=(R,C,U,\mathcal D)$. As previous described $G$ induces a stochastic process $\{R_t\}_t$ in the topological space $\{R'\in\mathbb R^l_+:C(R')=C(R)\}$. 
We say that a game $\mathcal G=(R,C,U,\mathcal D)$ is non-subsidizing if $p^\star=p_a$. In general, this equilibrium is not computable in polynomial time.
\end{definition}
\textbf{Observation}: The stochastic process $\{R_k\}_k$ is in fact a Markov process with state space
$\mathcal S=\{R\in\mathbb R_{\geq0}^l: C(R)=C(R_0)\}$ since a strategic player optimization problem just depends on its utility function, the current state and hers endowments. 
\begin{proposition}\label{prop:manifold} If $C:\mathbb R^l_+\rightarrow\mathbb R$ is a smooth map with boundaries, then $\mathcal S$ is a manifold with boundaries. If $C$ vanishes in the boundaries, then $\mathcal S$ is a manifold. 
\end{proposition}

\begin{proposition}\label{prop:smooth} Let $C$ and $U$ are smooth and strictly concave functions. Then for a given solution on a initial tuple $(\Delta_0,R_0)$ the function $\zeta(\cdot,\cdot)$ can be extended smoothly on all the domain.
\end{proposition}

\begin{definition}
A CFMMWD $G$ is \textbf{complete} if the map $\varphi_G:S\rightarrow\Delta^n$ defined as $R\mapsto \nabla C(R)/||\nabla C(R)||_1$ is differentiable, injective, image $\text{int}(\Delta^l)$ and with differential inverse. So, we can define the updating price stochastic process as
\begin{equation}
    p_{k+1}=\varphi(\varphi^{-1}(p_{k})+\zeta(\cdot,C,\varphi^{-1}(p_k))).
\end{equation}
we denote by $p^\star_s\in\Delta^l$ the \textbf{stochastic equilibrium price} defined as a solution of
\begin{equation}
    \mathbb E_{\Delta\sim\mathcal D}[\varphi\left(\varphi^{-1}(p^\star_s)+\zeta(\Delta,C,\varphi^{-1}(p^\star_s))\right)]=p_s^\star.
\end{equation}
\end{definition} 
\textbf{Counterexample}: Observe that in general $p^\star_s$ does not exist. For example, consider the CFMM game with $U(x,y)=C(x,y)=xy$ and $\mathcal D$ is given by $\Pr[(\Delta_1,\Delta_2)=(1,0)]=1$. Then, clearly we have that $R_k\rightarrow (0,1)$ as $k\rightarrow+\infty$. Since the image of $\varphi_G$ is in $(0,1)$, we have that $p^\star_s$ does not exist.
\begin{theorem}\label{theorem:existence} Let $G$ be a  complete CFMMWD with $p^\star>0$ and $\mathcal D$ with bounded support, then $p^\star_s$ exists.
\end{theorem}

\begin{proposition}\label{prop:convergence} Let  $\mathcal D$ be distribution of endowments, utility function $U$ with unique Walrasian equilibrium $p^\star$. Let $C:\mathbb R^l\rightarrow \mathbb R$ be convex and differentiable function and $R\in \mathbb R^l$ such that $\nabla C(R)=p^\star$. Then, it holds that the sequence of stochastic prices $\{p^\star_s(\lambda)\}_{\lambda\in\mathbb R_{\geq0}}$ of the CFMMWD $G^\lambda=(\lambda R,C,U,\mathcal D)$ converge to $p^\star$.
\end{proposition}

\subsection{Welfare and average price}
We define the notion of social welfare in CFMM with Walrasian demands. We will prove that, locally, the CFMM that provides more social welfare is the constant sum market maker with price $p^\star$. However, globally, we will prove that this is in general not true, leaving as an open question which function $C$ optimizes the social welfare for a given set of players with endowments $\mathcal D$ and utility function $U$.
\begin{definition}
For a CFMMWD, we define the \textbf{average social welfare} as
\begin{equation*}
    \text{WF}(G)=\lim_{t\rightarrow+\infty}\frac{\mathbb E_{\Delta(t)\sim\mathcal D^t}[\sum_{k=1}^tU(\zeta(R_k,\Delta_k,C))]}{t}.
\end{equation*}
We define the \textbf{average quilibrium price} as
\begin{equation*}
    p_a=\lim_{t\rightarrow+\infty}\frac{\mathbb E_{\Delta(t)\sim\mathcal D^t}[\sum_{k=1}^t\nabla C(R_k)/||\nabla C(R_k)||]}{t}.
\end{equation*}
\end{definition}

In general $Wf(G)$ does not need to exist, however there are sufficiently good conditions where this term is well-defined. Before doing that, we prove that locally the bests CFMM in terms of social welfare (without subsidizing) is the constant sum market maker with price $p^\star$. 
\begin{proposition}\label{prop:localWf} Let $R$ be an initial liquidity, $U$ a concave utility function and $\Delta\in\mathbb R^l_+$. 
Let $\text{Co}(\mathbb R^l_+,\mathbb R_+)$ be the set of all smooth convex functions from $\mathbb R^l_+$ to $\mathbb R_+$. Let $p^\star$ be the WEP of the one-shoot model with utility function $U$ and distribution $\mathcal D$. Assume that $\text{supp } z(\cdot,p^\star)\subseteq [0,R]^l$. Then,
\begin{equation*}
    \underset{s.t. \nabla C(R)=p^\star}{\text{max}_{C\in\text{Co}(\mathbb R^l_+,\mathbb R_+)}} U(\zeta(\Delta,R,C))=U(\zeta(\Delta,p^\star))
\end{equation*}
and the curve that is realized is $C(v)=v\cdot p^\star$. So, if $\mathcal D$ a distribution over $\mathbb R^l_+$ and  $\text{supp } z(\cdot,p^\star)\subseteq [0,M]^l$, for some $M>0$, and $R$ is sufficiently big, we have that:
\begin{equation*} 
    \underset{s.t. \nabla C(R)=p^\star}{\text{max}_{C\in\text{Co}(\mathbb R^l_+,\mathbb R_+)}}\mathbb E[U(\zeta(\Delta,R,C))]=\mathbb E[U(\zeta(\Delta,p^\star))]
\end{equation*}
\end{proposition}

As we mentioned, the $Wf(G)$ does not necessarily exist, however, if the Markov process $\{R_k\}_k$ associated to the CFMM has a unique stationary distribution, then this value is well-defined. Moreover, the average of utilities not just converges in expectancy but also converges almost surely to $Wf(G)$.
\begin{proposition}[Existence of Wf]\label{prop:equivalence_eq} Let $G=(R,U,C,\mathcal D)$ be a CFMMWD  and $\mathcal M=\{R_k\}_{k\in\mathbb N}$ its associated Markov process. If $\mathcal M$ has a unique stationary distribution $\pi$, and $\mathbb E_{\pi\times\mathcal D}[|U(\zeta(R,\Delta,C)|]<+\infty$ then
\begin{equation*}
\text{Wf}(G)\underset{a.s.}{=}  \lim_{t\rightarrow+\infty}\frac{\sum_{k=1}^tU(\zeta(R_k,\Delta_k,C))}{t} 
\end{equation*}
\end{proposition}

Now, let's provide an example of the computation of a $Wf(G)$ for a specific game.  Assume that the CFMM has constant function $C(x,y)=x+y$ and initial liquidity $(R_1,R_2)$. Assume that the representative agent has a Cobb-Douglas utility function $U(x,y)=xy$. The distribution of endowments is given by $\mathcal B(1, 1/2)^2$. Each round $t$, the agent with endowment $\Delta=(\Delta_1,\Delta_2)$ wants to solve
\begin{maxi*}|s|
{\Delta'}{(\Delta_1+x)(\Delta_2+y)}
{}{}
\addConstraint{R_1+R_2+x+y=R_1+R_2}{}
\addConstraint{R_1+x\geq0,R_2+y\geq0}{}
\end{maxi*}
Easily, we obtain that if  $(\Delta_1,\Delta_2)$ is $(0,0)$ or $(\Delta_{max},\Delta_{max})$, the agent will not interact with the CFMM since he is already maximizing his utility function. On the other hand, if $(\Delta_1,\Delta_2)=(\Delta_{max},0)$ or $(0,\Delta_{max})$, if possible, the agent will execute the trade, depositing $\Delta_{\max}/2$ of one asset and removing $\Delta_{\max}/2$ of the other one. Then the reserves of the CFMM behave as the following Markov chain process:
\begin{equation*}
    R_{k+1} = R_k+Y
\end{equation*}
where $Y$ is the random variable with distribution, $\Pr[Y=(-1/2,1/2)]=1/4$, $\Pr[Y=(-1/2,1/2)]=1/4$ and $\Pr[Y=(0,0)]=1/2$.

Assume that $R_1,R_2$ are integers and $\Delta_{max}=1$. Computing the stationary distribution $\pi$, we obtain that the probability that the Markov process is in the boundary set $\{(R_1+R_2,0),(0,R_1+R_2)\}$ is $p_b=2/R_1+R_2$. By, proposition \ref{prop:equivalence_eq} we have that
\begin{equation}
    \text{Wf}(G) = \left (1-\frac{1}{R_1+R_2}\right)\cdot\text{Welfare of one-shoot}
\end{equation}

Similar to \ref{prop:localWf} one could think that the constant sum Market Maker provides more social welfare than all others CFMM, however this is not true.

\begin{theorem}\label{theorem:negativewf} There exists a utility function $U$ and a distribution of endowments $\mathcal D$ such that the constant sum market maker with normal vector $p^\star$ (i.e. $C(x)=p^\star\cdot x)$ does not maximize welfare. In other words, there exist $C'$ such that $Wf(U,R,C')> Wf(U,R,C).$
\end{theorem}

\begin{proposition}\label{prop:welfare_plane} For a $\mathcal D$ distribution of endowments and utility function $U$ and Walrasian equilibrium price $p^\star$, we have that for all $R>0$.
\begin{equation*}
    \lim_{\lambda\rightarrow+\infty}Wf(G^\lambda)=Wf(\mathcal D,U)
\end{equation*}
where $G^\lambda = (\lambda R,C(x)=x\cdot p^\star,U,\mathcal D)$.
\end{proposition}

\begin{proposition}\label{prop:equiv_price} If the Markov process associated to $ G$ have a stationary distribution $\pi$, then we have that the average equilibrium price holds:
\begin{align*}
   \mathbb E_{R\sim\pi}\left[\frac{\nabla C(R)}{||\nabla C(R)||_1}\right]\underset{a.s}{=}\lim_{t\rightarrow+\infty}\frac{\sum_{k=1}^t\nabla C(R_k)/||\nabla C(R_k)||_1}{t}
    =\lim_{t\rightarrow+\infty}\frac{\sum_{k=1}^tp_k}{t}.
\end{align*}
 
\end{proposition}

\begin{theorem}\label{theorem:price_convergence} Assuming that $C$ is a convex complete CFMM an $\mathcal D$ is a distribution with bounded support. Then sequence of average prices of the CFMMWD $G^\lambda=(\lambda R,C,U,\mathcal D)$ converges to $p^\star$. In other words, if the liquidity is sufficiently large, we have that the oracle price of the CFMM converges to the Walrasian equilibrium price.
\end{theorem}
The proof of the theorem can be found in the appendix. Now, we will provide an example. Assume that players have Coubb-douglas utility function and that the distribution of endowments is given by $\mathcal U[0,1]^2$. Then, by symmetry, one can easily prove that $p_a=p^\star=(1,1)$. We simulated the CFMM with Walrasian demand with Uniswap V2 curvature and $t=2\cdot 10^6$. The initial reserves are set $(950,1050)$ and $(1000,1000)$ respectively.
\begin{figure}[!h]
\centering
\begin{subfigure}{.5\textwidth}
  \centering
  \includegraphics[width=.9\linewidth]{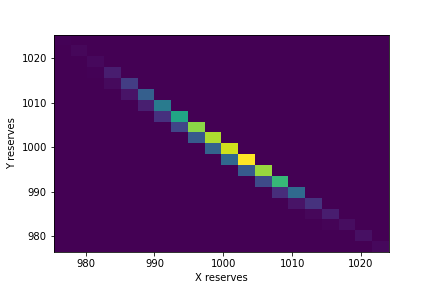}
  \caption{Reserves heat map}
  \label{fig:sub1}
\end{subfigure}%
\begin{subfigure}{.5\textwidth}
  \centering
  \includegraphics[width=.9\linewidth]{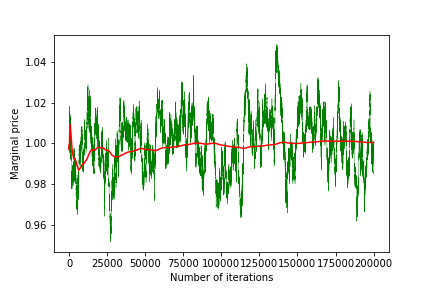}
  \caption{Prices and average price}
  \label{fig:sub2}
\end{subfigure}
\caption{Simulation of CFMM with Walrasian demand $(R,C=xy,U=xy,\mathcal U[0,1]^2)$ and initial reserves (1000,1000)}
\label{fig:test}
\end{figure}
\begin{figure}[!h]
\centering
\begin{subfigure}{.5\textwidth}
  \centering
  \includegraphics[width=.9\linewidth]{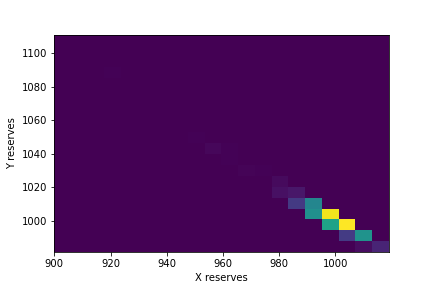}
  \caption{Reserves heat map}
  \label{fig:sub1}
\end{subfigure}%
\begin{subfigure}{.5\textwidth}
  \centering
  \includegraphics[width=.9\linewidth]{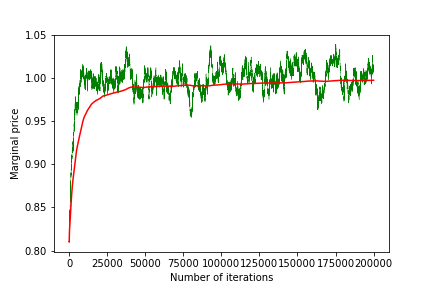}
  \caption{Prices and average price}
  \label{fig:sub2}
\end{subfigure}
\caption{Simulation of CFMM with Walrasian demand $(R,C=xy,U=xy,\mathcal U[0,1]^2)$ and initial reserves (900,1111.11)}
\label{fig:test}
\end{figure}

\section{Maximal extractable value in exchange economy}
In this section, we will try to lower bound the MEV in a Walrasian auctioneer mechanism with and without CFMM. 
Similar to the section \ref{one-shoot}, we assume that all players have a concave utility function $U$ and the endowments are drawn from a distribution $\mathcal D$. We assume that the players endowments and utility functions are truthfuly reported to the auctioneer. 
Assume the existence of a Walrasian auctioneer (in PBS this will be the builder) responsible for batching the transactions received and compute the Walrasian equilibrium price $p^\star$ and the output allocations. However, due to computation limitation per period of time (gas limit in Blockchain context) the number of transactions that the builder can settle is bounded by some number $N$. Moreover, we assume that this player share the same utility function and endowments distribution. 

Now, we will lower bound the MEV that the builder can extract in different scenarios. In this context, a transaction $t$ will be a tuple $(U,\Delta)$. First, we will assume that there is no bribing and that the probability of a non-builder transaction being added is uniformly random. That is, if the builder received $M$ transactions, then $\Pr[tx\in B]=\frac{N}{M}$.

\textbf{Builder}: In this context, the builder will be the responsible for executing the Walrasian equilibrium pricing algorithm for a given set of transactions $\mathcal T$. We will assume that the builder has the same utility function $U$ and initial endowments $\Delta$ drawn from the distribution $\mathcal D$. 
\begin{definition}
We say that the builder is \textbf{informed} if he has access to the transactions $\mathcal T$ (perfect signal). 
We say that a builder is a \textbf{censorship builder} if, for any given set of transactions $\mathcal T$, he can choose a subset of transactions $\mathcal T'\subseteq\mathcal T$ and compute the Walrasian allocation of this subset. 
We say that a builder \textbf{censorship-minimizer} if he constructs the Walrasian allocation of any subset with maximal cardinality bounded by $N$.
\end{definition}

Similarly to \cite{mazorra2022price}, we define the MEV as an optimization problem.
\begin{definition} Let $\Delta_b$ be the endowments of the builder and $\mathcal T$ be the set of transactions received. Then, the MEV of an informed and censorship builder $\text{MEV}(\Delta_b,\mathcal T)$ is:
\begin{maxi*}|s|
{B}{U(x_b(B\cup \{(U,\Delta_b)\})})
{}{}
\addConstraint{ B\subseteq \mathcal T, |B| \leq N-1}{}
\end{maxi*}
In the case that the builder minimize censorship, the MEV is defined analogously by adding the constraint that $|B| =\min\{|\mathcal T|+1,N-1\}$.
\end{definition}
In the limit cases, that is, if the number of transactions $M\rightarrow +\infty$ the MEV of the censorship-minimizer builder converges to the MEV of the censorship builder.
Observe that in general, the MEV is greater than zero. For example, assume that we have $l=2$, $\mathcal D$ is the distribution that $\Pr[(\Delta_1,\Delta_2)=(1,0)]=\Pr[(\Delta_1,\Delta_2)=(0,1)]$ and the players have the Coubb-douglas utility function $U(x,y)=xy$. Assume that the builder receives $3$ transactions. Two of the form $tx_1=(U(x,y)=xy,(1,0))$ and one of the form $tx_2=(U(x,y),(0,1))$. Moreover, assume that the builder has $(0,1)$ of initial endowments. Then, if the builder adds all transactions, its final endowments are $(1/2,1/2)$ obtaining a total utility of $1/4$. However, if the player censors the transaction $tx_2$, then his final endowments are $(1,1/2)$ leading to a utility of $1/2$.

\begin{proposition} When the builder is uninformed, individually rational, and risk-averse, then they want to add as many transactions as possible in a block. More formally, an uninformed builder has non-censoring as a dominant strategy. Moreover, we have that 
\begin{align*}
    \mathbb E(\text{MEV})&=Wf(\mathcal D,U)
\end{align*}
with $p$ following the distribution $X_{N}$ conditioned by one transaction being $\Delta$ and $p^\star$ being the Walrasian equilibrium. On the other hand, the expected utility of non-builder player is $\Pr[tx\in B]Wf(\mathcal D,U)+(1-\Pr[tx\in B])\mathbb E_{\Delta\sim\mathcal D}[U(\Delta)]$ and so the difference between the value extracted from a builder and a non-builder players is 
\begin{equation}
(1-\Pr[tx\in B])(Wf(\mathcal D,U)-\mathbb E_{\Delta\sim \mathcal D}[U(\Delta)]).
\end{equation}
\end{proposition}
\begin{proof}
By assumption made in preliminars, a risk averse agent will add as much transaction to reduce the risk of having a worst price than the Walrasian equilibrium price. Since, by the assumption,  the agent is truthful, we have that the builder will always add its transaction and remove another one. In this case, by the assumption its expected is $Wf(\mathcal D)$. 
\end{proof}
The last proposition gives us a lower bound on the one block MEV in a pure exchange economy with random endowments. Similar lower bounds emerge in frequent batch auctions such as \cite{mcmenamin2022fairtradex}. Moreover, we prove that with enough demand of exchange (that is $\Pr[tx\in B]<1$) the value extracted by the builder is strictly greater than the one extracted by a non-builder player. However, clearly if there is demand of exchange is because $Wf(\mathcal D)>\mathbb E_{\Delta\sim\mathcal D}[U(\mathcal D)]$. In other words,  a blockchain with unique validators per slot always have strict asymmetry of payoffs between validators and non-validators players if and only if there is sufficient demand of exchange through the chain.

\subsection{MEV in CFMM}

In this paper, we did not took into account the utility of the liquidity provider. Since, in general, any passive liquidity provider is weak against adverse selection. More formally, we have the following result.
Let $i\in\mathcal I$ be a player with initial endowments $R\in\mathbb R^k_+$, with utility function $U:\mathbb R^k_+\rightarrow \mathbb R$. There is no CFMM $C:\mathbb R^k_+\rightarrow \mathbb R$ such that the solutions of the optimization problems:

\begin{minipage}{.5\linewidth}
\begin{maxi*}|s|
{x}{U(x)}
{}{}
\addConstraint{c\cdot x = c\cdot R}{}
\end{maxi*}
\end{minipage}
\begin{minipage}{.3\linewidth}
\begin{maxi*}|s|
{x}{-c\cdot x}
{}{}
\addConstraint{C(x)=C(R)}{}
\end{maxi*}
\end{minipage}
\newline
$x^\star_1$ and $x^\star_2$ are have the same utility.
The first optimization problem models the best response of the player if there is a referential market with price vector $c$ (the Walrasian equilibrium price). That is, the liquidity provided by the LP is significantly small compared to the one provided by the off chain market.
The second optimization problem models how does the market react if the player posts a contract with initial reserves $R$ and CFMM $C$.

\textbf{Example}: Assume that the constant function is $C(x,y)=xy$ and the utility of the LP is $U=xy$. Can we find a utility function such that the optimization problem share the same solution?  Assume that the initial endowment is $R=(1,1)$. Wlog, we can assume that the vector price is of the form $c=(1,p)$. Then, the solution of the second optimization problem is realized in $R'=(\sqrt{p},1/\sqrt{p})$ and so the utility is $1$. On the other hand, the solution of he first optimization problem is given in $(\frac{1+p}{2},\frac{1+p}{2p})$ and so the utility is $\frac{(1+p)^2}{4p}$ and is strictly larger than $1$ for $p\not=1$. Therefore, we have that an LP with utility function $U=xy$ makes worse providing liquidity than rebalancing in the market.

\begin{proposition}\label{prop:LP_loss} Let $U$ be a strictly increasing concave utility function and $V_1,V_2\subseteq \mathbb R^l_{++}$ open sets. Then, for all complete CFMM $C$ it holds $U(x^\star_1(c,R))> U(x^\star_2(c,R))$ for every in $c\in V_1$ and initial reserves $R\in V_2$ such that $\nabla C(R)\not=c$.  
\end{proposition}
\begin{figure}[!h]
    \centering
    \includegraphics[scale=0.15]{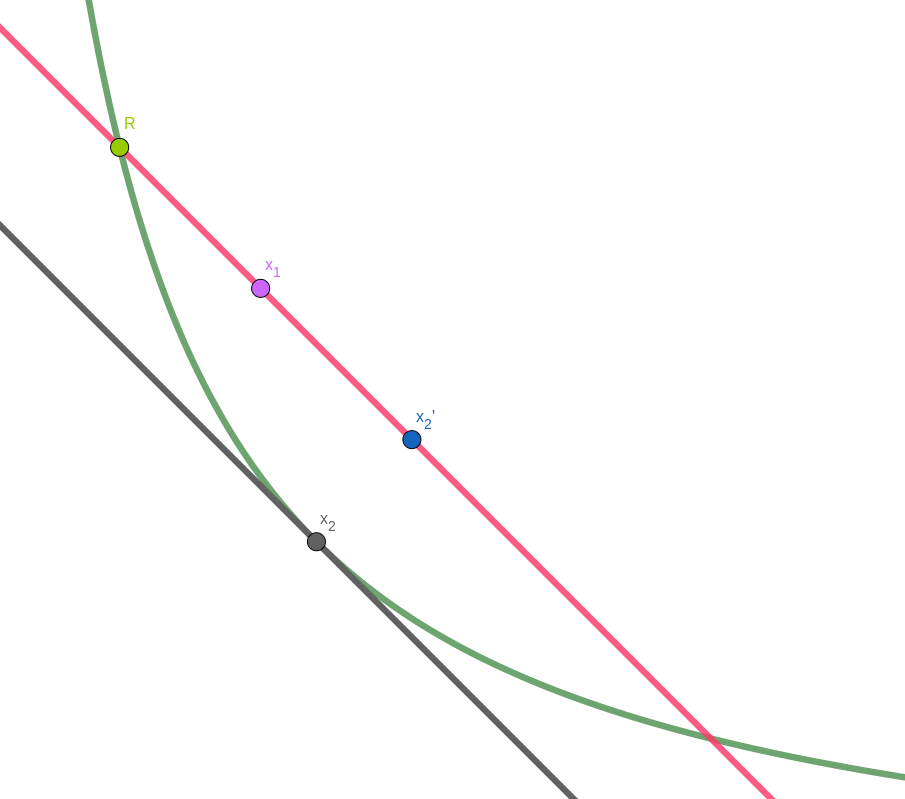}
    \caption{Geometric proof of LP loss}
    \label{fig:my_label}
\end{figure}
In other words, there is no CFMM that maximizes the trader's utility and the utility of the LP providers at the same time. Or, in other words, the LP positions are always exposed to adversarial selection, even with agents with Walrasian demands.

Now, we will discuss the expected MEV generated by the CFMM. For a given CFMM with reserves $R$. If validators' endowments are drawn randomly from the set of endowments $\mathcal D$, then the expected MEV is $\mathbb E[U(\zeta(\Delta,R,C)]$. Observe that in this case, the reserves that maximize the welfare are the ones that maximize the MEV. An interesting question is which reserves (or point in the feasible set) minimizes the MEV. In presence of another off chain market maker and  players that just values the non-risk asset the MEV is minimized when the marginal price of assets coincide off chain and in the CFMM (no arbitrage condition, see \cite{angeris2021constant}). However, this is not true with players with Walrasian demands. More specifically, there are cases where the CFMM being in the Walrasian equilibrium price maximizes MEV but also there are games where the CFMM being in the Walrasian equilibrium price minimizes the MEV

\textbf{Example}: We have the same distribution of endowments as previous examples and utility function $U(x,y)=\log(x)+\log(y)$. Clearly, in this case the MEV is not minimized in reserves $(R,R)$ since taking the reserves $\underset{x\rightarrow+\infty}{\lim}\mathbb E[U(\zeta(\Delta,C,(R^2/x,x)))]=-\infty$.

\section{Future Work}
While in the special case of the model we analysed the MEV does not reduce social welfare, this is not in general the case. Characterising the loss in social welfare from MEV in a more general model where the validator has visibility into the content of transactions. This naturally motivates the design of variations of CFMMs that have higher social welfare. Another open problem is under what conditions does the average price of the CFMM converges to the Walrasian equilibrium and if so how does the MEV impact on the rate convergence.

\printbibliography
\appendix
\section{Appendix(Work in Progress)}\label{appendix}

\textbf{Proof \ref{prop:manifold}}: Clearly, we have that $\nabla C(x)\not=0$ for all $x\in \mathbb R^l_{++}$. In other words, all non-trivial elements are regular. By implicit function theorem, we have that $C^{-1}(R)$ is a manifold of dimension $l-1$.$\square$

\textbf{Proof \ref{prop:smooth}} First, using Lagrangian and the implicit function theorem, we can prove that the function $\zeta$ can be defined locally and is smooth. We can extend it globally since the maximum always exist and by continuity of the function
\begin{equation*}
    L(\Delta,\Lambda,R,\mu)=U(\Delta)-\mu(C(R+\Delta-\Lambda)-C(R)).\square
\end{equation*}

\textbf{Proof \ref{theorem:existence}}:
Let's consider the map $F:p\mapsto \mathbb E_{\Delta\sim\mathcal D}[\varphi\left(\varphi^{-1}(p^\star_s)+\zeta(\Delta,C,\varphi^{-1}(p^\star_s))\right)]$. Since $\varphi,\varphi^{-1}$ and $\zeta(\Delta,C,R)$ are continuous, we deduce that $F$ is continuous by using the convergence dominated theorem. 
Now, we will prove $\varepsilon>0$ such that $N_\varepsilon=\Delta^l_\varepsilon=\{x\in\Delta^l:x\geq\varepsilon\textbf{1}\}$ holds $F(N_\varepsilon)\subseteq N_\varepsilon$. Assume otherwise, then exist a sequence of points $\{x_n\}_n$ such that $d(F(x_n),\partial \Delta^l)\rightarrow0$ (where $d$ denotes the Euclidean distance and $\partial \Delta^l$ the frontier of the $l$ dimensional simplex). Since $\Delta^l$ is compact, wlog we can assume that $F(x_n)\rightarrow q$ for some $q\in\Delta^l$ and that $F(x_n)_i\rightarrow0$ and $x_j\rightarrow0$ for some $i$ and $j$ (taking subsequences). Since $d(F(x_n),\partial \Delta^l)\rightarrow 0$, we deduce that $q\in\partial \Delta^l$. However, since the support of the endowments $\mathcal D$ is finite, we have that $i=j$ (is deduced from the fact that if $C$ is complete). But if $i=j$, we would deduce that $p_i=0$, leading to a contradiction. 
Therefore, exist $N_\varepsilon$ such that $F(N_\varepsilon)\subseteq N_\varepsilon$. Since $N_\varepsilon\cong \Delta^l$, using the Brouwer theorem \cite{hatcher2005algebraic}, we have that exist $x^\star$ such that $F(x^\star)=x^\star$.$\square$

\textbf{Proof \ref{prop:convergence}}
 Since the Walrasian equilibrium of the pure exchange economy is unique, we have that is enough to show that  $\mathbb E[z(\Delta, \lim_{\lambda\rightarrow+\infty} p^\star_s(\lambda))]=0$. Since the map $p\mapsto \mathbb E[z(\Delta,p)]$ is continuous, we have that $\mathbb E[z(\Delta, \lim_{\lambda\rightarrow+\infty} p^\star_s(\lambda))]=\lim_{\lambda\rightarrow+\infty} \mathbb E[z(\Delta, p^\star_s(\lambda))]$. On the other hand, since $\nabla C(\lambda R)/||C(\lambda R)||_1=p^\star$, we have that $z(\Delta,C,\lambda R)\underset{\lambda\rightarrow+\infty}{\longrightarrow} z(\Delta,p^\star)$ uniformly at $\Delta$. So,
\begin{align*}
    \mathbb E[z(\Delta, \lim_{\lambda\rightarrow+\infty} p^\star_s(\lambda))]
    &= \lim_{\lambda\rightarrow+\infty} \mathbb E[z(\Delta, p^\star_s(\lambda))]\\
    &= \lim_{\lambda\rightarrow+\infty} \mathbb E[z(\Delta, C, p^\star_s(\lambda))]\\
    &=0.\square
\end{align*}

\textbf{Proof \ref{prop:localWf}}
Clearly, using that $\text{supp } z(\cdot,p^\star)\subseteq [0,M]^l$ we have that if $C(v)=v\cdot p^\star$ then $\mathbb E[U(\zeta(\Delta,R,C))]=\mathbb E[U(\zeta(\Delta,p^\star))]$. Moreover, the feasible set associated to $C_{p^\star}(v)=v\cdot p^\star$ with reserves $R$ is $E=\{R'\in\mathbb R^l_+:R'\cdot p^\star\geq R\cdot p^\star\}$. Now, We have to prove that  $\underset{s.t. \nabla C(R)=p^\star}{\text{max}_{C\in\text{Co}(\mathbb R^l_+,\mathbb R_+)}}\mathbb E[U(\zeta(\Delta,R,C))]\leq\mathbb E[U(\zeta(\Delta,p^\star))]$. 
If we denote by $S$ the feasible set induced by $C$. we have that $S\subseteq E$. This holds from the fact that $\nabla C(R)=p^\star$ and $C$ is concave. By the upper bound constraints on $\Delta$, we have that $z(\Delta,R,C)$ holds all the optimization problem constraints and therefore we deduce the result. The second part is deduced immediately by taking expectancies.$\square$

\textbf{Proof \ref{prop:equivalence_eq}}: If $\mathcal M$ has a stationary distribution, then, since $U$ is bounded, we deduce it by the large law of numbers, for more details see \cite{hillier1967introduction}. $\square$

\textbf{Proof \ref{theorem:negativewf}}: Take Utility function $U_n(x,y)=\log(x+1/n)+\log(y+1/n)$. Assume endowments hold the distribution $\mathcal D$ defined as $\Pr[(\Delta_1,\Delta_2)=(1,0)]=\Pr[(\Delta_1,\Delta_2)=(0,1)]=1/2$. Let $\pi_n$ be the stationary distribution followed by the Markov process $(R,U_n,C=x+y,\mathcal D)$. Observe that in this CFMMWD, the utility of the player after the trade is bounded by $2\log(1+1/n)$. Therefore, we have that
\begin{equation*}
    Wf(U_n,R,C)\leq (1-\pi_b/2)2\log(1+1/n)+\frac{\pi_b}{2}(\log(1+1/n)+\log(1/n)),
\end{equation*}
where $\pi_b=\text{Probability of reserves being in the boundary}$.
Observe that 
\begin{equation*}
    \lim_{n\rightarrow\infty}  Wf(U_n,R,C)=-\infty
\end{equation*}
Now, we will prove that there exist a $C$ such that the  CFMMWD $G=(R,U_n,C,\mathcal D)$ holds $Wf_n(G)\geq L$ for some $L$. 
Take the constant product market maker $C(x,y)=xy$. Now, for an endowment, $(0,1)$ we can compute the best trade. Assume that the current reserves of the CFMM are $R_x,R_y$. The player will trade a quantity $0\leq b\leq 1$ such that maximizes $G(b)(1-b)$ where $G(b)=\frac{-k}{R_y+b}+R_x$.
Computing it, we have that $b=-R_y+\sqrt{R_y(R_y+1)}$. On the other hand, the $\pi$ distribution of the reserves is uniquely determineted for the $x$ reserves $R_x$. 
One can prove that exist $x_{min}$ and $x_{max}$ such that $\pi(x)\leq 1/2^{1/x}$ for $x\in[0,x_{min}]$ and $1-\pi(x)\leq 1/2^x$ for $x\in[x_{max},+\infty]$. 
Therefore, we deduce that $\int_{-\infty}^\infty(\log(x)+\log(g(x))\pi(x)dx$ exits and is finite, deducing the result. $\square$

\textbf{Proof \ref{prop:welfare_plane}}:
The proof is deduced from the fact that the probability of being in the border tends to zero as $\lambda$ tends to infinity and the proposition \ref{prop:equivalence_eq}.$\square$

\textbf{Proof Sketch \ref{prop:equiv_price}}:
The first equality is deduced similar to \ref{prop:equivalence_eq}, the second one by definition.

\textbf{Proof \ref{theorem:price_convergence}}:

\begin{lemma} Let $G^\lambda=(\lambda R,C,U,\mathcal D)$ be a sequence of CFMM with Walrasian demand with $C(x)=p\cdot x$ with $p\not=p^\star$. Assume that the support of the endowments is bounded. Then:
\begin{enumerate}
    \item $G^\lambda$ is Ergodic for all $\lambda\geq0$, and so, we can consider the stationary distribution $\pi^\lambda$.
    \item Let $\mathcal S$ be the state space of $G^1$ and $A\subseteq \text{int}(S)$ closed set, then $\lim_{\lambda\rightarrow+\infty}\pi^\lambda\left[x\in \lambda A\right]=0$.
\end{enumerate}
\end{lemma}
\begin{proof}
\begin{enumerate}
    \item  Since $\mathcal S^\lambda$ is compact, we have that every collection of measures are tight. Moreover, the Markov operator associated to the Markov process is Feller. Therefore, the statement is deduced by the Krylov-Bogolioubov theorem.
    \item This is equivalent to proof that positive random walks with negative drift have a stationary distribution. Since, we can approximate the Markov chain by sufficiently closer Markov chains. 
\end{enumerate}
The theorem then is deduced from approximating the CFMM by piece-wise linear CFMM and proving that the probability that the point is in the piece with price $p\not=p^\star$ is zero.
\end{proof}

\textbf{Proof \ref{prop:LP_loss}}: Let $R$ be the initial reserves of the CFMM. Take $c$ such that $C(R)\not=c$. Let $R'$ be the reserves after the trade $x_2(c,R)$. Taking Lagrangian, we can show that $\nabla C(R')=c$. Since $C$ is convex, we have that
\begin{equation*}
    C(R)\leq C(R')+\nabla C(R')\cdot(R-R')
\end{equation*}
Since $C(R')=C(R)$, we have that $c R\geq c R'$. Therefore, we deduce that there exist $\lambda\geq1$ such that $\lambda x_2\in\{x:xc=x_1c\}$. Now we will proof that $\lambda>1$. Assume that $\lambda=1$, then we would have that $cR=cR'$. Since $C$ is concave, we have that the set $\mathcal S=\{R:C(R)\geq C(R)\}$ is convex and so the line $\{x:xc=x_2c\}$ is tangent. This implies that $\nabla C(R)=c$, leading to a  contradiction, therefore $\lambda>1$. Since $U$ is strictly increasing, we have that $U(\lambda x_2(c,R))> U(x_2(c,R))$. On the other hand, $\lambda x_2$ holds the constraints of the first optimization problem, therefore, $U(x_1(c,R))> U(\lambda x_2(c,R))$. $\square$

\end{document}